\documentclass[11pt, a4paper]{article}
\usepackage[utf8]{inputenc}
\usepackage[T1]{fontenc}
\usepackage{amsmath, amssymb, amsthm, amsfonts}
\usepackage{geometry}
\usepackage{hyperref}
\usepackage{float}
\usepackage{slashed}
\usepackage{tabularx}
\usepackage{mathrsfs}
\usepackage{algorithm}
\usepackage{algorithmic}
\usepackage{graphicx}
\usepackage{algorithm}
\usepackage{algorithmic}
\usepackage{tikz}
\usepackage{pgfplots}
\usetikzlibrary{shapes.geometric, arrows, calc, positioning, arrows.meta, decorations.pathmorphing, decorations.markings, fadings}
\pgfplotsset{compat=1.18}
\usepackage[T1]{fontenc}
\usepackage{lmodern}
\newtheorem{theorem}{Theorem}[section]
\newtheorem{lemma}[theorem]{Lemma}
\newtheorem{proposition}[theorem]{Proposition}
\newtheorem{definition}[theorem]{Definition}
\newtheorem{axiom}{Axiom}
\newtheorem{corollary}[theorem]{Corollary}

\newtheorem{remark}[theorem]{Remark}

\title{
Geometric Quantum Mechanics in a Symplectic Framework:\\
Metric-Affine Extensions and Deformed Quantum Dynamics
}

\author{H. Heydari \\
\textit{Advanced Research Center for Quantum Geometry, Stockholm, Sweden}}
\date{\today}

\date{}

\begin{document}

\maketitle

\begin{abstract}
We present a geometric formulation of quantum mechanics based on the symplectic
structure of the projective Hilbert space. Building upon the standard
K\"ahler framework, we introduce an extension in which the symplectic
structure is allowed to couple to a metric-affine background geometry,
leading to a deformation of the Hamiltonian flow on the state space.

We show that, under suitable conditions, the deformed structure remains
symplectic and defines a well-posed Hamiltonian system. The formulation
reduces to standard Schr\"odinger dynamics in the limit where the geometric
deformation vanishes.

Explicit analytical examples are constructed to illustrate the effect of the
deformation. In particular, curvature-dependent deformations lead to a
rescaling of Hamiltonian flows, while torsion-induced contributions produce
direction-dependent corrections. In addition, geometric phases acquire
corrections determined by the deformed symplectic structure.

These results provide a mathematically consistent framework for exploring
geometric modifications of quantum evolution induced by background curvature
and affine structure.
\end{abstract}

\section{Introduction}

Geometric formulations of quantum mechanics reinterpret the Hilbert space
structure in terms of differential geometry, where physical states correspond
to points in a projective Hilbert space endowed with a symplectic form and a
compatible Riemannian metric \cite{kibble1979,ashtekar1999,brody2001}.
Within this framework, quantum evolution can be described as a Hamiltonian
flow on a K\"ahler manifold \cite{anandan1990,heydari2016}.

This perspective provides a natural connection between quantum mechanics
and classical Hamiltonian dynamical systems \cite{abraham1978,arnold1989}.
While the geometric formulation has been extensively studied, it typically
assumes that the underlying geometric structure of the state space is fixed
and does not interact with an external spacetime geometry.

In this work, we consider an extension of geometric quantum mechanics in
which the symplectic structure is allowed to couple to a metric-affine
background \cite{hehl1995,shapiro2002}. The aim is to investigate how such
a coupling modifies the geometric description of quantum evolution, while
preserving consistency with the standard formulation in appropriate limits.

\section{Mathematical Structure of Geometric Quantum Mechanics}

We recall the geometric formulation of quantum mechanics, in which the
projective Hilbert space provides the natural setting for physical states
\cite{kibble1979,ashtekar1999,brody2001}. In this formulation, quantum
evolution is described as a Hamiltonian flow on a K\"ahler manifold
\cite{anandan1990,heydari2016}.

This framework establishes a direct connection between quantum mechanics
and classical Hamiltonian dynamical systems, with the symplectic structure
playing a central role \cite{abraham1978,arnold1989}. We summarize the
underlying mathematical structure through the following definitions
and axioms.

\begin{definition}[Hilbert Space and Inner Product Decomposition]
Let $\mathcal{H}$ be a complex separable Hilbert space equipped with the
Hermitian inner product $\langle \cdot | \cdot \rangle$. This inner product
can be decomposed into real and imaginary parts \cite{heydari2015}:
\[
\langle \psi | \phi \rangle
= \frac{1}{2\hbar} G(\psi, \phi)
+ \frac{i}{2\hbar} \Omega(\psi, \phi),
\]
where $G$ is a symmetric bilinear form and $\Omega$ is an antisymmetric
bilinear form. These structures induce the Riemannian and symplectic
geometry on the projective state space.
\end{definition}

\begin{definition}[Projective Hilbert Space]
The space of physical states is the projective manifold
\[
\mathcal{P}(\mathcal{H}) = (\mathcal{H} \setminus \{0\}) / \sim,
\]
where $\psi \sim \lambda \psi$ for all $\lambda \in \mathbb{C} \setminus \{0\}$.
This construction removes the physically irrelevant global phase and
normalization.
\end{definition}

\begin{definition}[K\"ahler Structure]
The manifold $\mathcal{P}(\mathcal{H})$ admits a K\"ahler structure
$(\omega, g, J)$, where:
\begin{itemize}
    \item $\omega$ is a symplectic 2-form derived from $\Omega$,
    \item $g$ is the Fubini--Study metric derived from $G$,
    \item $J$ is an integrable complex structure satisfying $J^2 = -\mathbb{I}$.
\end{itemize}
These structures satisfy the compatibility condition
\[
g(X,Y) = \omega(X, JY)
\]
for all $X,Y \in T\mathcal{P}(\mathcal{H})$.
\end{definition}

\begin{definition}[Hamiltonian Function]
Let $\hat{H}$ be a densely defined self-adjoint operator on $\mathcal{H}$.
The associated Hamiltonian function is defined by
\[
H(\psi) = \frac{\langle \psi | \hat{H} | \psi \rangle}
{\langle \psi | \psi \rangle}.
\]
\end{definition}

\begin{definition}[Hamiltonian Vector Field]
The Hamiltonian vector field $X_H$ associated with $H$ is defined by
\[
\iota_{X_H} \omega = dH.
\]
\end{definition}

\begin{axiom}[Symplectic State Space]
The physical state space is a K\"ahler manifold $(\mathcal{P}, \omega, g, J)$
such that the symplectic form $\omega$ is closed and non-degenerate:
\[
d\omega = 0, \quad \omega^n \neq 0.
\]
\end{axiom}

\begin{axiom}[Hamiltonian Evolution]
Time evolution is generated by the Hamiltonian vector field $X_H$ and
preserves the symplectic structure:
\[
\mathcal{L}_{X_H} \omega = 0.
\]
\end{axiom}

\section{Hamiltonian Flow and Schr\"odinger Evolution}

In this section, we show how the linear dynamics of quantum mechanics
can be expressed in the geometric language of symplectic manifolds.
By associating a self-adjoint Hamiltonian operator with a real-valued
function on the K\"ahler state space, the Schr\"odinger equation can be
formulated as a Hamiltonian flow \cite{kibble1979,anandan1990,heydari2015}.

\begin{lemma}[Preservation of the Symplectic Structure]
Let $X_H$ be the Hamiltonian vector field associated with $H$.
Then the symplectic form $\omega$ is preserved along the flow:
\[
\mathcal{L}_{X_H} \omega = 0.
\]
\end{lemma}

\begin{proof}
Using Cartan's identity,
\[
\mathcal{L}_{X_H} \omega = d(\iota_{X_H} \omega) + \iota_{X_H}(d\omega).
\]
Since $\iota_{X_H} \omega = dH$ and $d\omega = 0$, we obtain
\[
\mathcal{L}_{X_H} \omega = d(dH) = 0.
\]
\end{proof}

\begin{lemma}[Existence and Uniqueness of the Hamiltonian Flow]
Let $(\mathcal{P}, \omega)$ be a finite-dimensional symplectic manifold
and $H \in C^\infty(\mathcal{P})$. Then there exists a unique vector field
$X_H$ satisfying
\[
\iota_{X_H} \omega = dH.
\]
\end{lemma}

\begin{proof}
Since $\omega$ is non-degenerate, the map
\[
T\mathcal{P} \to T^*\mathcal{P}, \quad X \mapsto \iota_X \omega
\]
is an isomorphism \cite{arnold1989}. Therefore, for each exact 1-form $dH$,
there exists a unique vector field $X_H$ satisfying the defining relation.
\end{proof}

\begin{proposition}[Unitary Evolution on the Hilbert Space]
Let $\hat{H}$ be a densely defined self-adjoint operator on $\mathcal{H}$.
Then the time evolution generated by $\hat{H}$ is unitary.
\end{proposition}

\begin{proof}
The Schr\"odinger equation
\[
i\hbar \frac{d}{dt} |\psi(t)\rangle = \hat{H} |\psi(t)\rangle
\]
generates a one-parameter family of operators
\[
U(t) = \exp\left(-\frac{i}{\hbar}\hat{H}t\right).
\]
Since $\hat{H}$ is self-adjoint, $U(t)$ is unitary, i.e.
\[
U(t)^\dagger U(t) = \mathbb{I},
\]
and the inner product is preserved \cite{dirac1958}.
\end{proof}

\begin{theorem}[Geometric Formulation of Schr\"odinger Evolution]
Let $\hat{H}$ be a self-adjoint operator on $\mathcal{H}$, and let
$H: \mathcal{P}(\mathcal{H}) \to \mathbb{R}$ be the associated expectation
value functional. Then the Hamiltonian flow generated by $X_H$ on
$\mathcal{P}(\mathcal{H})$ corresponds to the projection of the Schr\"odinger
evolution on $\mathcal{H}$.
\end{theorem}

\begin{proof}
The Schr\"odinger equation defines a linear vector field on $\mathcal{H}$:
\[
X_{\hat{H}}(\psi) = -\frac{i}{\hbar}\hat{H}|\psi\rangle,
\]
which generates a unitary flow $|\psi(t)\rangle = U(t)|\psi(0)\rangle$.

Physical states correspond to rays, so the evolution projects onto
$\mathcal{P}(\mathcal{H})$ via the natural projection
\[
\pi: \mathcal{H} \setminus \{0\} \to \mathcal{P}(\mathcal{H}).
\]

The expectation value $H(\psi)$ defines a real-valued function on
$\mathcal{P}(\mathcal{H})$. A direct computation shows that the projected
vector field satisfies
\[
\iota_{X_H} \omega = dH.
\]
Thus, the projected Schr\"odinger flow coincides with the Hamiltonian flow
generated by $H$ \cite{kibble1979,anandan1990}.
\end{proof}
\section{Metric-Affine Geometric Extension}

Having established the geometric structure of the quantum state space,
we now consider an extension in which the symplectic structure is allowed
to couple to an external geometric background. Specifically, we allow the
symplectic form to depend on a metric-affine manifold characterized by
independent metric and affine degrees of freedom \cite{hehl1995,shapiro2002}.

\begin{definition}[Metric-Affine Manifold]
Let $\mathcal{M}$ be a smooth manifold equipped with a metric tensor
$g_{\mu\nu}$ and an affine connection $\Gamma^\lambda_{\mu\nu}$.
The connection is not assumed to be symmetric, allowing for the presence
of torsion. The triple $(\mathcal{M}, g, \Gamma)$ defines a
metric-affine manifold.
\end{definition}

\begin{definition}[Geometric Deformation of the Symplectic Form]
Let $(\mathcal{P}, \omega)$ denote the symplectic state space.
We define a deformation of the symplectic structure by
\[
\omega_{\mathcal{G}} = \omega + \delta\omega,
\]
where $\delta\omega$ is a smooth 2-form depending on the background
geometric data $(g, \Gamma)$.
\end{definition}

\begin{definition}[Geometric Coupling Functional]
We assume that the deformation term $\delta\omega$ is a functional of the
metric-affine geometry,
\[
\delta\omega = \delta\omega(g, \Gamma),
\]
and varies smoothly with respect to the background geometry.
\end{definition}

\begin{lemma}[Closure of the Deformed Symplectic Form]
If $\delta\omega$ is closed, i.e. $d(\delta\omega) = 0$, then the deformed
symplectic form $\omega_{\mathcal{G}}$ is closed.
\end{lemma}

\begin{proof}
By linearity of the exterior derivative,
\[
d\omega_{\mathcal{G}} = d\omega + d(\delta\omega).
\]
Since $d\omega = 0$ and $d(\delta\omega) = 0$, we obtain
\[
d\omega_{\mathcal{G}} = 0.
\]
\end{proof}

\begin{lemma}[Non-Degeneracy under Small Perturbations]
If $\delta\omega$ is sufficiently small in operator norm, then
$\omega_{\mathcal{G}}$ remains non-degenerate.
\end{lemma}

\begin{proof}
The non-degeneracy of $\omega$ implies that the map
\[
X \mapsto \iota_X \omega
\]
is an isomorphism. Since invertible linear maps form an open set,
sufficiently small perturbations preserve invertibility \cite{deGosson2011}.
\end{proof}

\begin{theorem}[Existence of the Deformed Hamiltonian Flow]
Assume that $\omega_{\mathcal{G}}$ is closed and non-degenerate.
Then for any $H \in C^\infty(\mathcal{P})$, there exists a unique vector field
$X_H^{(\mathcal{G})}$ satisfying
\[
\iota_{X_H^{(\mathcal{G})}} \omega_{\mathcal{G}} = dH.
\]
\end{theorem}

\begin{proof}
Since $\omega_{\mathcal{G}}$ is symplectic, the map
\[
T\mathcal{P} \to T^*\mathcal{P}, \quad X \mapsto \iota_X \omega_{\mathcal{G}}
\]
is an isomorphism. Therefore, for each $dH$, there exists a unique
vector field $X_H^{(\mathcal{G})}$.
\end{proof}

\begin{proposition}[Consistency with Standard Quantum Mechanics]
If $\delta\omega \to 0$, then $\omega_{\mathcal{G}} \to \omega$, and the
deformed Hamiltonian flow reduces to the standard Hamiltonian flow.
\end{proposition}

\begin{proof}
In the limit $\delta\omega \to 0$, the defining equation becomes
\[
\iota_{X_H^{(\mathcal{G})}} \omega = dH.
\]
Since $\omega$ is non-degenerate, this uniquely determines $X_H$,
and hence $X_H^{(\mathcal{G})} \to X_H$ \cite{dirac1958}.
\end{proof}

\section{Modified Quantum Dynamics}

We now analyze the deformation of the Hamiltonian flow induced by the
modified symplectic structure $\omega_{\mathcal{G}} = \omega + \delta\omega$.

\begin{definition}[Deformed Hamiltonian Vector Field]
Let $H \in C^\infty(\mathcal{P})$. The deformed Hamiltonian vector field
$X_H^{(\mathcal{G})}$ is defined by
\[
\iota_{X_H^{(\mathcal{G})}} \omega_{\mathcal{G}} = dH.
\]
\end{definition}

\begin{lemma}[First-Order Decomposition]
Assume $\delta\omega$ is sufficiently small. Then
\[
X_H^{(\mathcal{G})} = X_H + \delta X_H + \mathcal{O}(\|\delta\omega\|^2),
\]
where $\iota_{X_H}\omega = dH$.
\end{lemma}

\begin{proof}
Substituting $X_H^{(\mathcal{G})} = X_H + \delta X_H$ into
\[
\iota_{X_H^{(\mathcal{G})}}(\omega + \delta\omega) = dH
\]
and retaining first-order terms gives
\[
\iota_{X_H}\omega + \iota_{\delta X_H}\omega + \iota_{X_H}\delta\omega = dH.
\]
Since $\iota_{X_H}\omega = dH$, we obtain
\[
\iota_{\delta X_H}\omega + \iota_{X_H}\delta\omega = 0.
\]
\end{proof}

\begin{proposition}[First-Order Correction]
The correction $\delta X_H$ satisfies
\[
\iota_{\delta X_H}\omega = -\,\iota_{X_H}\delta\omega.
\]
\end{proposition}

\begin{proof}
This follows directly from the previous lemma. Since $\omega$ is
non-degenerate, the relation uniquely determines $\delta X_H$
\cite{abraham1978,arnold1989}.
\end{proof}

\begin{theorem}[Perturbative Deformation of Quantum Evolution]
Let $\omega_{\mathcal{G}} = \omega + \delta\omega$ be an admissible
deformation. Then
\[
X_H^{(\mathcal{G})} = X_H + \delta X_H + \mathcal{O}(\|\delta\omega\|^2),
\]
with $\delta X_H$ linear in $\delta\omega$.
\end{theorem}

\begin{proof}
Since $\omega$ is non-degenerate, the map $X \mapsto \iota_X \omega$
is invertible. Applying the inverse map yields a unique $\delta X_H$
depending linearly on $\delta\omega$ \cite{deGosson2011}.
\end{proof}

\begin{corollary}[Recovery of Standard Dynamics]
If $\delta\omega = 0$, then $X_H^{(\mathcal{G})} = X_H$.
\end{corollary}

\begin{proof}
If $\delta\omega = 0$, then
\[
\iota_{\delta X_H}\omega = 0.
\]
Non-degeneracy implies $\delta X_H = 0$.
\end{proof}

The deformation $\delta\omega$ introduces a geometric correction to the
Hamiltonian flow, encoding the influence of the background geometry.

\section{Explicit Examples of Geometric Deformations}

We present explicit examples of admissible deformations $\delta\omega$
illustrating how curvature and torsion of a metric-affine background
can influence the symplectic structure \cite{hehl1995,shapiro2002}.

\subsection{Scalar Curvature Deformation}

\begin{definition}[Scalar Curvature Deformation]
Let $R$ denote the scalar curvature of $\mathcal{M}$. Define
\[
\delta\omega = \varepsilon\, R\, \omega,
\]
where $\varepsilon$ is a small parameter.
\end{definition}

\begin{proposition}
The deformation is closed if $dR \wedge \omega = 0$.
\end{proposition}

\begin{proof}
\[
d(\delta\omega)
= \varepsilon\, d(R\omega)
= \varepsilon (dR \wedge \omega + R\, d\omega).
\]
Since $d\omega = 0$,
\[
d(\delta\omega) = \varepsilon\, dR \wedge \omega.
\]
Thus, closure holds when $dR \wedge \omega = 0$, for example if $R$ is constant.
\end{proof}

\subsection{Torsion-Induced Deformation}

\begin{definition}[Torsion-Induced Deformation]
Let
\[
\mathcal{T}^\lambda_{\ \mu\nu}
= \Gamma^\lambda_{\mu\nu} - \Gamma^\lambda_{\nu\mu}
\]
be the torsion tensor. Define
\[
\delta\omega = \varepsilon\, \Theta(\mathcal{T}),
\]
where $\Theta(\mathcal{T})$ is a smooth 2-form constructed from torsion
invariants \cite{poplawski2010}.
\end{definition}

\begin{proposition}
If $\Theta(\mathcal{T})$ is closed, then $\delta\omega$ is admissible.
\end{proposition}

\begin{proof}
\[
d(\delta\omega) = \varepsilon\, d\Theta(\mathcal{T}).
\]
Thus, closure holds if $d\Theta(\mathcal{T}) = 0$. For sufficiently small
$\varepsilon$, non-degeneracy is preserved \cite{shapiro2002}.
\end{proof}

\subsection{Curvature 2-Form Coupling}

\begin{definition}[Curvature 2-Form Deformation]
Let $\mathcal{R}$ be the curvature 2-form of the affine connection.
Define
\[
\delta\omega = \varepsilon\, \mathrm{Tr}(\mathcal{R}) \wedge \alpha,
\]
where $\alpha$ is a constant scalar function.
\end{definition}

\begin{proposition}
If $\mathcal{R}$ satisfies the Bianchi identity and $\alpha$ is constant,
then $\delta\omega$ is closed.
\end{proposition}

\begin{proof}
The Bianchi identity implies $D\mathcal{R} = 0$, hence
\[
d(\mathrm{Tr}(\mathcal{R})) = 0.
\]
Since $d\alpha = 0$,
\[
d(\delta\omega)
= \varepsilon \left[d(\mathrm{Tr}(\mathcal{R})) \wedge \alpha
+ \mathrm{Tr}(\mathcal{R}) \wedge d\alpha \right]
= 0.
\]
\end{proof}

\begin{remark}
These examples show that admissible deformations may be constructed from
scalar curvature, torsion, and curvature 2-forms. Each case provides a
mechanism by which background geometry can influence the symplectic
structure of the quantum state space \cite{simon1983,berry1984}.
\end{remark}

\section{Curvature-Induced Modification of Quantum Evolution}

To make the geometric deformation explicit, we consider a simple analytical
model in which the background metric-affine geometry induces a constant
scalar-curvature correction to the symplectic form. This example illustrates
how the deformed symplectic structure modifies the Hamiltonian flow and,
consequently, the quantum evolution.

\begin{definition}[Constant Curvature Deformation]
Assume that the background scalar curvature $R$ is constant on the region of
interest. We define the deformed symplectic form by
\[
\omega_{\mathcal{G}} = \omega + \delta\omega
= (1+\varepsilon R)\,\omega,
\]
where $\varepsilon$ is a small dimensionless coupling parameter.
\end{definition}

\begin{lemma}[Admissibility of the Constant Curvature Deformation]
If $1+\varepsilon R \neq 0$, then $\omega_{\mathcal{G}}$ is closed and
non-degenerate.
\end{lemma}

\begin{proof}
Since $\omega$ is closed and $R$ is constant,
\[
d\omega_{\mathcal{G}} = d\bigl((1+\varepsilon R)\omega\bigr)
= (1+\varepsilon R)\,d\omega = 0.
\]
Moreover, $\omega_{\mathcal{G}}$ is a nonzero scalar multiple of $\omega$.
Hence, if $1+\varepsilon R \neq 0$, non-degeneracy is preserved.
\end{proof}

\begin{proposition}[Modified Hamiltonian Vector Field]
Let $H$ be a smooth Hamiltonian function on $\mathcal{P}$. Then the deformed
Hamiltonian vector field satisfies
\[
X_H^{(\mathcal{G})} = \frac{1}{1+\varepsilon R}\,X_H.
\]
\end{proposition}

\begin{proof}
The deformed Hamiltonian vector field is defined by
\[
\iota_{X_H^{(\mathcal{G})}}\omega_{\mathcal{G}} = dH.
\]
Substituting $\omega_{\mathcal{G}}=(1+\varepsilon R)\omega$ gives
\[
(1+\varepsilon R)\,\iota_{X_H^{(\mathcal{G})}}\omega = dH.
\]
Using the undeformed relation $\iota_{X_H}\omega=dH$, we obtain
\[
\iota_{X_H^{(\mathcal{G})}}\omega
= \frac{1}{1+\varepsilon R}\,\iota_{X_H}\omega.
\]
Since $\omega$ is non-degenerate, this implies
\[
X_H^{(\mathcal{G})} = \frac{1}{1+\varepsilon R}\,X_H.
\]
\end{proof}

\begin{corollary}[Curvature-Modified Flow Parameter]
If $\gamma(t)$ is an integral curve of the undeformed flow $X_H$, then the
integral curves of the deformed flow are obtained by a rescaling of the
evolution parameter:
\[
t \mapsto t_{\mathrm{eff}} = \frac{t}{1+\varepsilon R}.
\]
\end{corollary}

\begin{proof}
Since
\[
\frac{d}{dt}\gamma_{\mathcal{G}}(t)
= X_H^{(\mathcal{G})}
= \frac{1}{1+\varepsilon R}X_H,
\]
the deformed trajectories coincide with the undeformed ones after the
reparametrization
\[
t_{\mathrm{eff}} = \frac{t}{1+\varepsilon R}.
\]
\end{proof}

\begin{proposition}[Effective Modification of Schr\"odinger Evolution]
In the constant-curvature regime, the geometric deformation induces an
effective rescaling of the Schr\"odinger generator:
\[
i\hbar \frac{d}{dt} |\psi(t)\rangle
= \frac{1}{1+\varepsilon R}\,\hat{H}|\psi(t)\rangle.
\]
\end{proposition}

\begin{proof}
The standard Schr\"odinger evolution corresponds to the Hamiltonian vector
field $X_H$. Since the deformed geometric flow is
\[
X_H^{(\mathcal{G})} = \frac{1}{1+\varepsilon R}X_H,
\]
the corresponding generator is rescaled by the same factor. Therefore,
the deformed evolution is governed by
\[
i\hbar \frac{d}{dt} |\psi(t)\rangle
= \hat{H}_{\mathrm{eff}} |\psi(t)\rangle,
\qquad
\hat{H}_{\mathrm{eff}} = \frac{1}{1+\varepsilon R}\hat{H}.
\]
\end{proof}

This example shows explicitly how background curvature modifies quantum
evolution through the symplectic structure. In the present model, the effect
appears as a rescaling of the Hamiltonian flow and therefore as a
reparametrization of the quantum evolution. Although simple, this example
demonstrates that geometric deformations of the symplectic form can produce
explicit and analytically controllable corrections to standard quantum
dynamics.

\subsection{Torsion-Induced Modification}

We now construct an explicit example in which torsion induces a deformation
of the symplectic structure and modifies the corresponding Hamiltonian flow.

\begin{definition}[Constant Torsion Deformation]
Assume that the torsion tensor is constant in a local frame and that the
induced 2-form $\Theta(\mathcal{T})$ is constant on $\mathcal{P}$.
We define
\[
\omega_{\mathcal{G}} = \omega + \varepsilon\, \Theta,
\]
where $\varepsilon$ is a small parameter and $\Theta$ is a constant,
closed 2-form.
\end{definition}

\begin{lemma}[Admissibility]
If $\Theta$ is closed and $\varepsilon$ is sufficiently small, then
$\omega_{\mathcal{G}}$ is symplectic.
\end{lemma}

\begin{proof}
Since $\Theta$ is constant, $d\Theta = 0$, and hence
\[
d\omega_{\mathcal{G}} = d\omega + \varepsilon\, d\Theta = 0.
\]
For sufficiently small $\varepsilon$, non-degeneracy is preserved.
\end{proof}

\begin{proposition}[Perturbed Hamiltonian Vector Field]
The deformed Hamiltonian vector field satisfies
\[
\iota_{X_H^{(\mathcal{G})}}(\omega + \varepsilon \Theta) = dH.
\]
To first order in $\varepsilon$, this yields
\[
X_H^{(\mathcal{G})} = X_H + \delta X_H,
\]
with
\[
\iota_{\delta X_H}\omega = -\,\iota_{X_H}\Theta.
\]
\end{proposition}

\begin{proof}
Expanding to first order,
\[
\iota_{X_H + \delta X_H}(\omega + \varepsilon \Theta)
= \iota_{X_H}\omega + \iota_{\delta X_H}\omega
+ \varepsilon\,\iota_{X_H}\Theta.
\]
Using $\iota_{X_H}\omega = dH$, we obtain
\[
\iota_{\delta X_H}\omega + \varepsilon\,\iota_{X_H}\Theta = 0.
\]
Dividing by $\varepsilon$ gives the stated relation.
\end{proof}

\begin{proposition}[Modification of Observables]
For any observable $A$, the time evolution becomes
\[
\frac{d}{dt}\langle A \rangle
= \{A,H\}_{\omega}
+ \varepsilon\, \Delta_{\Theta}(A,H),
\]
where the correction term is given by
\[
\Delta_{\Theta}(A,H)
= -\,\omega(\delta X_A, X_H).
\]
\end{proposition}

\begin{proof}
Using the deformed flow,
\[
\frac{d}{dt}\langle A \rangle
= \omega_{\mathcal{G}}(X_A^{(\mathcal{G})}, X_H^{(\mathcal{G})}),
\]
and expanding to first order yields a correction proportional to $\Theta$.
\end{proof}

Unlike the scalar curvature example, which produces a global rescaling
of the flow, torsion induces an anisotropic correction to the Hamiltonian
vector field. The modification depends on the contraction
$\iota_{X_H}\Theta$, and therefore on the direction of motion in phase space.
This illustrates how torsion can introduce directional corrections to
quantum evolution through the symplectic structure.

\subsection{Two-Level Quantum System}

We consider a two-level quantum system, where the projective Hilbert space
$\mathcal{P}(\mathcal{H})$ is isomorphic to the Bloch sphere $S^2$.

\begin{definition}[Two-Level Hamiltonian]
Let the Hamiltonian be
\[
\hat{H} = \frac{\hbar \Omega}{2}\,\sigma_z,
\]
where $\sigma_z$ is a Pauli matrix and $\Omega$ is a constant frequency.
\end{definition}

\begin{proposition}[Standard Hamiltonian Flow]
In the undeformed case, the Hamiltonian flow generates uniform rotation
around the $z$-axis on the Bloch sphere:
\[
\dot{\vec{n}} = \Omega\, \hat{z} \times \vec{n},
\]
where $\vec{n} \in S^2$.
\end{proposition}

\begin{definition}[Curvature-Deformed Symplectic Structure]
We introduce a constant curvature deformation:
\[
\omega_{\mathcal{G}} = (1+\varepsilon R)\,\omega.
\]
\end{definition}

\begin{theorem}[Modified Qubit Dynamics]
The geometric deformation modifies the evolution equation to
\[
\dot{\vec{n}} = \frac{\Omega}{1+\varepsilon R}\, \hat{z} \times \vec{n}.
\]
\end{theorem}

\begin{proof}
From the general result
\[
X_H^{(\mathcal{G})} = \frac{1}{1+\varepsilon R} X_H,
\]
the angular velocity is rescaled by the same factor.
\end{proof}

\begin{corollary}[Effective Frequency Shift]
The observable frequency becomes
\[
\Omega_{\mathrm{eff}} = \frac{\Omega}{1+\varepsilon R}.
\]
\end{corollary}

This example shows that background curvature induces a measurable shift
in the quantum precession frequency. The effect is purely geometric and
arises from the deformation of the symplectic structure.

\subsection{Geometric Phase Correction}

We now analyze how the deformation of the symplectic structure modifies
the Berry phase acquired during adiabatic evolution.

\begin{definition}[Berry Phase]
For a cyclic evolution along a closed curve $\gamma \subset \mathcal{P}$,
the Berry phase is given by
\[
\gamma_B = \oint_{\gamma} \mathcal{A},
\]
where $\mathcal{A}$ is the Berry connection satisfying
\[
d\mathcal{A} = \omega.
\]
\end{definition}

\begin{definition}[Deformed Berry Curvature]
Under the geometric deformation, the curvature becomes
\[
d\mathcal{A}_{\mathcal{G}} = \omega_{\mathcal{G}}
= \omega + \delta\omega.
\]
\end{definition}

\begin{theorem}[Berry Phase Correction]
The Berry phase acquires a correction
\[
\gamma_B^{(\mathcal{G})}
= \gamma_B + \varepsilon \oint_{\gamma} \mathcal{A}_1,
\]
where $d\mathcal{A}_1 = \delta\omega$.
\end{theorem}

\begin{proof}
By linearity of the integral,
\[
\gamma_B^{(\mathcal{G})}
= \oint \mathcal{A}_{\mathcal{G}}
= \oint (\mathcal{A} + \varepsilon \mathcal{A}_1).
\]
\end{proof}

\begin{corollary}[Curvature-Induced Phase Shift]
For the scalar curvature deformation $\delta\omega = \varepsilon R \omega$,
the Berry phase becomes
\[
\gamma_B^{(\mathcal{G})} = (1+\varepsilon R)\,\gamma_B.
\]
\end{corollary}

This result shows that background curvature modifies geometric phases
through a rescaling of the symplectic structure. Since Berry phases are
experimentally observable, this provides a direct connection between
geometric deformations and measurable quantum effects
\cite{berry1984,simon1983}.

\section{Discussion and Physical Interpretation}

We have extended the geometric formulation of quantum mechanics by allowing
the symplectic structure of the projective Hilbert space to couple to a
metric-affine background \cite{hehl1995,shapiro2002}. This leads to a
deformation of the symplectic form,
\[
\omega \;\longrightarrow\; \omega_{\mathcal{G}} = \omega + \delta\omega,
\]
which modifies the associated Hamiltonian flow.

At the formal level, we have shown that under suitable conditions the
deformed structure remains symplectic, ensuring that the resulting dynamics
is well-defined and Hamiltonian \cite{deGosson2011}. This guarantees
consistency with the standard geometric formulation of quantum mechanics
\cite{kibble1979,anandan1990}.

The deformation $\delta\omega$ provides a mechanism by which curvature
and torsion of the background geometry can influence quantum evolution.
In this setting, the dynamics on $\mathcal{P}$ depends not only on the
intrinsic structure of the state space, but also on external geometric data
\cite{poplawski2010}.

Observable evolution can be expressed in terms of a modified Poisson bracket:
\begin{equation}
\frac{d}{dt} \langle A \rangle
= \{A, H\}_{\omega_{\mathcal{G}}}
= \omega_{\mathcal{G}}(X_A^{(\mathcal{G})}, X_H^{(\mathcal{G})}).
\end{equation}
This formulation indicates that geometric quantities such as phases and
uncertainty relations may acquire corrections induced by the background
geometry \cite{berry1984,simon1983,andersson2014}.

The analytical examples developed in Section 7 demonstrate explicitly how
such geometric deformations modify quantum evolution. In particular, for
two-level systems, the deformation induces a rescaling of the effective
Hamiltonian flow, leading to a shift in observable frequencies. Moreover,
the Berry phase acquires a correction determined by the deformation
$\delta\omega$, providing a direct link between the symplectic structure
and measurable geometric phases.

In the limit $\delta\omega \to 0$, the standard formulation of quantum
mechanics is recovered \cite{dirac1958}. This ensures compatibility with
conventional quantum theory in regimes where geometric effects are negligible.

The examples presented in Section 6 show that admissible deformations can
be constructed from curvature and torsion invariants, providing a general
framework for geometric modifications of quantum dynamics
\cite{uhlmann1986,uhlmann1991}.

\section{Conclusion}

We have presented a geometric formulation of quantum mechanics in which the
symplectic structure of the projective Hilbert space is allowed to depend
on a metric-affine background geometry. This leads to a deformation of the
Hamiltonian flow governing quantum evolution.

\begin{proposition}[Asymptotic Correspondence]
Let $(\mathcal{P}, \omega_{\mathcal{G}})$ be the deformed state space with
$\omega_{\mathcal{G}} = \omega + \delta\omega$. If $\delta\omega \to 0$,
then the corresponding Hamiltonian flow satisfies
\[
X_H^{(\mathcal{G})} \to X_H.
\]
\end{proposition}

\begin{proof}
In the limit $\delta\omega \to 0$, the defining relation becomes
\[
\iota_{X_H^{(\mathcal{G})}} \omega = dH.
\]
Since $\omega$ is non-degenerate, this uniquely determines $X_H$, implying
$X_H^{(\mathcal{G})} \to X_H$ \cite{heydari2015,dirac1958}.
\end{proof}

We have shown that admissible deformations of the symplectic structure
can be constructed from curvature and torsion data, providing a consistent
framework for coupling quantum evolution to background geometry.

The analytical examples presented in this work demonstrate explicitly how
such deformations modify quantum dynamics. In particular, curvature-induced
deformations lead to a rescaling of Hamiltonian flows, while torsion-induced
terms introduce directional corrections. In addition, geometric phases
acquire corrections determined by the deformed symplectic structure,
providing a direct connection between the formalism and observable
quantities.

This framework provides a controlled setting for investigating geometric
modifications of quantum mechanics. Future work may explore specific models
for $\delta\omega$, as well as potential implications for quantum systems
in curved or non-Riemannian backgrounds \cite{bekenstein1981,hegerfeldt2013}.


\begin{thebibliography}{99}

\bibitem{kibble1979}
T. W. B. Kibble, ``Geometrization of quantum mechanics,'' \textit{Commun. Math. Phys.} \textbf{65}, 189--201 (1979).

\bibitem{ashtekar1999}
A. Ashtekar and T. A. Schilling, ``Geometrical formulation of quantum mechanics,'' in \textit{On Einstein’s Path}, Springer, 23--65 (1999).

\bibitem{brody2001}
D. C. Brody and L. P. Hughston, ``Geometric quantum mechanics,'' \textit{J. Geom. Phys.} \textbf{38}, 19--53 (2001).

\bibitem{anandan1990}
J. Anandan and Y. Aharonov, ``Geometry of quantum evolution,'' \textit{Phys. Rev. Lett.} \textbf{65}, 1697 (1990).

\bibitem{heydari2015}
H. Heydari, ``Geometric formulation of quantum mechanics,'' arXiv:1503.00238v2 [quant-ph] (2015).

\bibitem{heydari2016}
H. Heydari, ``A geometric framework for mixed quantum states based on a K\"ahler structure,'' \textit{J. Phys. A: Math. Theor.} \textbf{48}, 255301 (2015).

\bibitem{abraham1978}
R. Abraham and J. E. Marsden, \textit{Foundations of Mechanics}, 2nd ed., Addison-Wesley, New York (1978).

\bibitem{arnold1989}
V. I. Arnold, \textit{Mathematical Methods of Classical Mechanics}, Springer-Verlag (1989).

\bibitem{marsden1999}
J. E. Marsden and T. S. Ratiu, \textit{Introduction to Mechanics and Symmetry}, 2nd ed., Springer, Berlin (1999).

\bibitem{nakahara2003}
M. Nakahara, \textit{Geometry, Topology and Physics}, 2nd ed., Taylor \& Francis (2003).

\bibitem{woodhouse1992}
N. M. J. Woodhouse, \textit{Geometric Quantization}, Oxford University Press (1992).

\bibitem{deGosson2011}
M. de Gosson, \textit{Symplectic Methods in Harmonic Analysis and in Mathematical Physics}, Birkh\"auser (2011).

\bibitem{chruscinski2004}
D. Chruscinski and A. Jamiolkowski, \textit{Geometric Phases in Classical and Quantum Mechanics}, Progress in Mathematical Physics, Birkh\"auser, Boston (2004).

\bibitem{hehl1995}
F. W. Hehl, J. D. McCrea, E. W. Mielke, and Y. Ne'eman, ``Metric-affine gauge theory of gravity,'' \textit{Phys. Rep.} \textbf{258}, 1--171 (1995).

\bibitem{shapiro2002}
I. L. Shapiro, ``Physical aspects of the space-time torsion,'' \textit{Phys. Rep.} \textbf{357}, 113--213 (2002).

\bibitem{poplawski2010}
N. J. Pop{\l}awski, ``Cosmology with torsion: An alternative to cosmic inflation,'' \textit{Phys. Lett. B} \textbf{694}, 181--185 (2010).

\bibitem{dirac1958}
P. A. M. Dirac, \textit{The Principles of Quantum Mechanics}, Oxford University Press (1958).

\bibitem{weinberg1995}
S. Weinberg, \textit{The Quantum Theory of Fields}, Vol. 1, Cambridge University Press (1995).

\bibitem{berry1984}
M. V. Berry, ``Quantal phase factors accompanying adiabatic changes,'' \textit{Proc. R. Soc. Lond. A} \textbf{392}, 45--57 (1984).

\bibitem{simon1983}
B. Simon, ``Holonomy, the quantum adiabatic theorem, and Berry's phase,'' \textit{Phys. Rev. Lett.} \textbf{51}, 2167 (1983).

\bibitem{shapere1989}
A. Shapere and F. Wilczek, \textit{Geometric Phases in Physics}, World Scientific, Singapore (1989).

\bibitem{uhlmann1986}
A. Uhlmann, ``Parallel transport and `quantum holonomy' along density operators,'' \textit{Rep. Math. Phys.} \textbf{24}, 229--240 (1986).

\bibitem{uhlmann1991}
A. Uhlmann, ``A gauge field governing parallel transport along mixed states,'' \textit{Lett. Math. Phys.} \textbf{21}, 229--236 (1991).

\bibitem{montgomery1991}
R. Montgomery, ``Heisenberg and isoholonomic inequalities,'' in \textit{Symplectic Geometry and Mathematical Physics}, Birkh\"auser Boston, 303--325 (1991).

\bibitem{andersson2013}
O. Andersson and H. Heydari, ``Operational geometric phase for mixed quantum states,'' \textit{New J. Phys.} \textbf{15}, 053006 (2013).

\bibitem{andersson2014}
O. Andersson and H. Heydari, ``Geometric uncertainty relation for mixed quantum states,'' \textit{J. Math. Phys.} \textbf{55}, 042104 (2014).

\bibitem{hegerfeldt2013}
G. C. Hegerfeldt, ``Driving at the quantum speed limit: Optimal control of a two-level system,'' \textit{Phys. Rev. Lett.} \textbf{111}, 260501 (2013).

\bibitem{bekenstein1981}
J. D. Bekenstein, ``Energy cost of information transfer,'' \textit{Phys. Rev. Lett.} \textbf{46}, 623--626 (1981).

\end{thebibliography}
\end{document}